\documentclass[10pt, a4paper, conference]{IEEEtran}
 \pdfoutput=1
\usepackage[english,swedish]{babel}			
\usepackage[noadjust]{cite}				
\usepackage[utf8]{inputenc}				
\usepackage{graphicx}					
\usepackage{listings} 					
\usepackage{lscape} 					
\usepackage{amsmath}					
\usepackage{amssymb}					
\usepackage{mathrsfs}					
\usepackage{amsthm}
\usepackage{url} 						
\usepackage{nameref}
\usepackage{longtable}					
\usepackage[usenames,dvipsnames]{color}		
\usepackage{tikz}						
\usepackage{enumerate}				
\usepackage{algpseudocode}				
\usepackage{algorithm}					
\usepackage{epstopdf}					

\usetikzlibrary{arrows}

\usepackage{flushend}

\tikzstyle{int}=[draw, fill=white, minimum size=2em]
\tikzstyle{init} = [pin edge={<-,thin,black}]
\tikzstyle{line} = [pin edge={-,thin,black}]
\tikzstyle{fran} = [pin edge={->,thin,black}]
\tikzstyle{Idea}=[draw, fill=lightgray,  fill opacity=0.75, text opacity = 1, minimum size=2em]

\theoremstyle{definition}
\newtheorem{definition}{Definition}
\newtheorem{lemma}{Lemma}
\theoremstyle{plain}
\newtheorem{theorem}{Theorem}

\newtheorem{proposition}{Proposition}

\theoremstyle{remark}
\newtheorem*{remark}{Remark}

\DeclareMathOperator{\sinc}{sinc}
\DeclareMathOperator{\diag}{diag}

\DeclareMathOperator{\rect}{rect}

\DeclareMathOperator{\Tr}{Tr}

\newcommand{\R}{\mathbb{R}}
\newcommand{\C}{\mathbb{C}}
\newcommand{\N}{\mathbb{N}}

\newcommand{\X}{\mathbf{X}}

\newcommand{\E}{\mathbf{E}}
\newcommand{\gv}{\mathcal{N}}
\newcommand{\lp}{\left(}
\newcommand{\rp}{\right)}
\newcommand{\intinf}{\int_{-\infty}^\infty}

\newcommand{\dd}{\,\mathrm{d}} 
\newcommand{\Ltwo}{\mathcal{L}_2(\R)} 
\newcommand{\Lone}{\mathcal{L}_1([-\pi,\, \pi])} 

\title{Time Localization and Capacity of Faster-Than-Nyquist Signaling}
\date{}
\author{
\IEEEauthorblockN{Ather Gattami}
\IEEEauthorblockA{Bitynamics Research\\ Stockholm, Sweden\\
Email: atherg@gmail.com}
\and
\IEEEauthorblockN{Emil Ringh}
\IEEEauthorblockA{Ericsson Research\\ Stockholm, Sweden\\
Email: emil.ringh@ericsson.com}
\and
\IEEEauthorblockN{Johan Karlsson}
\IEEEauthorblockA{KTH Royal Institute of Technology\\ Stockholm, Sweden\\
Email: johan.karlsson@math.kth.se}
}

\begin{document}
\selectlanguage{english}
\maketitle

\begin{abstract}
In this paper, we consider communication over the bandwidth limited analog white Gaussian noise channel using \textit{non-orthogonal} pulses. 
In particular, we consider non-orthogonal transmission by signaling samples at a rate higher than the Nyquist rate. 
Using the faster-than-Nyquist (FTN) framework, Mazo showed that one may transmit symbols carried by $\sinc$ pulses at a higher rate than that dictated by Nyquist without loosing bit error rate. However, as we will show in this paper, such pulses are not necessarily well localized in time. In fact, assuming that signals in the FTN framework are well localized in time, one can construct a signaling scheme that violates the Shannon capacity bound. We also show directly that FTN signals are in general not well localized in time. Therefore, the results of Mazo do not imply that one can transmit more data \textit{per time unit} without degrading performance in terms of error probability.

We also consider FTN signaling in the case of pulses that are different from the $\sinc$ pulses. We show that one can use a precoding scheme of low complexity to remove the inter-symbol interference. This leads to the possibility of increasing the number of transmitted samples per time unit and compensate for spectral inefficiency due to signaling at the Nyquist rate of the non $\sinc$ pulses. 
We demonstrate the power of the precoding scheme by simulations. 
\end{abstract}

\section{Introduction}

\subsection{Background and previous work}
In 1949, Shannon \cite{shannon1949} presented his famous result on the capacity of the additive white Gaussian noise (AWGN) channel for band limited signals. The result is based on communication using orthogonal pulses (to avoid inter-symbol interference) which corresponds to transmission at the Nyquist sampling rate.

In \cite{mazo:FTN}, Mazo considered the \textit{uncoded} transmission case for binary symbols, where the sample rate is faster than the one dictated by the Nyquist sampling theorem, so called \textit{faster-than-Nyquist} (FTN) signaling. The minimum Euclidean distance between pairs of binary signals was used as a performance measure, and it was shown that one can transmit the signal pulses faster than the Nyquist frequency without decreasing the Euclidian distance between any two signals. The limit for such rate increase is known as the Mazo limit: $\rho \approx 0.802$,  and was derived in \cite{mazo:min_dist, hajela:on_computing}. Recently, a series of work has explored FTN signaling, see \cite{rusek:FTN}  and the references therein. 
This work has also been extended to two dimensions, the second dimension being the frequency domain \cite{rusek:two_dim}. The signals are packed tighter also in the frequency domain, possibly introducing interference between previously uncorrelated subcarriers. However, this inter-carrier interference does not affect the reliability of the signaling with a certain packing density, and with the use of an optimal detector the error rates would remain unchanged. Also, in \cite{rusek:2009} it was shown that FTN can achieve the maximum capacity for a given pulse (as opposed to orthogonal transmission). This work also considers the use of FTN for achieving higher capacity for non-$\sinc$ pulses when the code alphabet is finite.

The original derivation of Shannon's capacity formula \cite{shannon1949} is based on transmission of uncorrelated $\sinc$ pulses.  These signals have infinite support in the time domain and hence, as pointed out by Wyner in \cite{wyner1966capacity},  ``the idea of transmission {\em rate} [using band limited pulses] has, at best, a limited meaning.''  To treat this problem in a rigorous manner, Wyner proposed several physically consistent models with corresponding coding theorems, thereby justifying Shannon's capacity formula \cite{wyner1966capacity}. Also, in the FTN framework, the concept of transmission rate is problematic and there is no guarantee that a signal consisting of a linear combination of pulses centered at a given time interval has its energy localized in the vicinity of that interval. In fact, as we will see in Section \ref{sec:timeloc}, one can easily construct examples where this is not the case. Another problem in the FTN framework is that the algorithms used for detection and estimation suffers from high complexity, rendering NP-hard problems in general \cite{verdu,proakis}.

\subsection{Contribution}
We consider communication over the bandwidth limited analog Gaussian white noise channel using
\textit{non-orthogonal} pulses. 
In particular, we consider non-orthogonal transmission by signaling samples at a rate higher than the Nyquist rate. 
Using the faster-than-Nyquist (FTN) framework, Mazo showed that one may transmit symbols carried by $\sinc$ pulses at a higher rate than that dictated by Nyquist without loosing bit error rate. However, we show in this paper that such pulses are not necessarily well localized in time. In fact, assuming that signals in the FTN framework are well localized in time, we show that one can construct a signaling scheme that violates the Shannon capacity bound. We also show directly that FTN signals are in general not well localized in time. Thus, it's not physically correct to talk about bits per \textit{second}, as the energy of non-orthogonal signals may not be well localised in time, as opposed to orthogonal $\sinc$ pulses, for instance. 
Therefore, the results of Mazo \cite{mazo:FTN} do not imply that one can transmit more data \textit{per time unit} without degrading performance in terms of error probability.

We go on and consider FTN signaling in the case of pulses that are different from the $\sinc$ pulses. We show that one may use a precoding scheme of low complexity in order
to remove the inter-symbol interference. This leads to the possibility of increasing the number of transmitted samples per time unit (keeping average energy per time unit bounded by some constant) and compensating for spectral efficiency losses due to signaling at the Nyquist rate, and so, achieve the Shannon capacity when the same energy is spent for transmitting with the ideal $\sinc$ pulses. We demonstrate the power of the precoding scheme by simulations.


\section{Preliminaries and Properties of Non-Orthogonal Pulses}

To start, we introduce some basic notation. Let $\N$ denote the set of natural numbers $\{0, 1, 2, ...\}$, 
$\R$ the set of real numbers, and
$\C$ the set of complex numbers. The $n\times n$ identity matrix is denoted by
$I_n$.
$X \sim \gv(0, \X)$ denotes that $X$ is a Gaussian variable with
$\E\{X\} = 0$ and $\E\{XX^\intercal\} = \X$. We use $ \X\succ 0$ ($ \X\succeq 0$) to denote that the matrix $\X$ is positive definite (semidefinite). $\X^\intercal$ is the transpose of a matrix and $\X^*$ is the conjugate transpose. The empty set is denoted by $\emptyset$
and $\mathbf{int} (\Omega)$ denotes the interior of the set $\Omega$.
The floor function $\lfloor x \rfloor$ denotes the greatest integer less than or equal to $x$. 
We define the normalized $\sinc$ function as
$$
\sinc(x) = \frac{\sin{\pi x}}{\pi x}.
$$
We also define the function
\begin{equation*}
	\text{rect}(x) = \left\{
		\begin{aligned}
			1 &~~~ |x| \leq \frac{1}{2} \\
			0 &~~~ |x| > \frac{1}{2}.
		\end{aligned} \right.
\end{equation*}


\begin{definition}[Signal Space]
Let $\Ltwo$ be the Hilbert space of functions $f:\R \rightarrow \C$ that are square integrable, endowed with the scalar product
$$
\langle f, g\rangle = \int_{-\infty}^\infty f(t)\overline{g(t)} dt
$$
and squared norm
$$
\|f\|^2 = \int_{-\infty}^\infty |f(t)|^2 dt < \infty.
$$
\end{definition}

\begin{definition}[Fourier Transform]
The Fourier transform of a function $f\in \Ltwo$ is given by
$$
f(t) \xrightarrow{\mathcal{F}}  F(\omega) = \intinf f(t) e^{-i\omega t} dt.
$$
\end{definition}

\begin{definition}
$\Lone$ is the Banach space of functions $f:\R \rightarrow \C$ that are absolutely integrable, with the norm
$$
 \|f\| = \int_{-\pi}^\pi |f(t)| dt < \infty  .
$$
\end{definition}

\begin{definition}[Gaussian White Noise Process \cite{gallager}]
The stochastic process $Z(t)$ is a white Gaussian noise process if $Z(t)$ has a constant spectral
density over all frequencies. 
\end{definition}

\begin{proposition}
\label{gv}
Let $Z(t)$ be a zero mean white Gaussian noise process with spectral density $\frac{N_0}{2}$. 
For $g_k, g_\ell\in \Ltwo$, define
$$
Z_k = \langle Z, g_k \rangle = \int_{-\infty}^\infty Z(t)\overline{g_k(t)} dt
$$  
and
$$
Z_\ell = \langle Z, g_\ell \rangle = \int_{-\infty}^\infty Z(t)\overline{g_\ell(t)} dt.
$$  
Then, $Z_k$ and $Z_\ell$ are zero mean Gaussian variables with covariance 
$$
\E\{Z_k \bar{Z}_\ell\} = \frac{N_0}{2} \langle g_k, g_\ell \rangle = \frac{N_0}{2} \int_{-\infty}^\infty g_k(t) \overline{g_\ell(t)}  dt.
$$
\end{proposition}
\begin{proof}
See \cite{gallager}.
\end{proof}

\begin{proposition}[Shannon Capacity of the Discrete Memory-less Gaussian Channel]
Let $\{Z_k\}_{k=1}^n$ be a set of independent Gaussian variables with $Z_k\sim \gv(0,\frac{N_0}{2})$
and consider the Gaussian channel with input sequence of samples $\{X_k\}_{k=1}^n$ and output
$$
Y_k = X_k + Z_k, \mbox{ for } k = 1, \ldots, n.
$$
Suppose that
$$
\sum_{k=1}^n \E |X_k|^2 \leq n E_s.
$$
Then, the capacity of the channel is given by
$$
C = \log_2 \lp1 + \frac{2E_s}{N_0} \rp ~~~ \textup{bits/sample}
$$
and is achieved for $X_k\sim \gv(0,E_s)$.

\end{proposition}
\begin{proof}
See \cite{shannon:48}.
\end{proof}

\begin{definition}[Toeplitz matrix]\label{def:toeplitz_mat}
A matrix $T_n$ is called Toeplitz if it has the following form
$$T_n = \left( \begin{array}{c c c c c}c_{0} & c_{-1} & c_{-2} & \dots & c_{-(n-1)}\\
c_1 & c_0 & c_{-1} & \dots & c_{-(n-2)}\\
c_2 & c_1 & c_{0} & \dots & c_{-(n-3)}\\
\vdots & \vdots & \vdots & \ddots & \vdots\\
c_{n-1} & c_{n-2} & c_{n-3} & \dots & c_{0}
\end{array}\right).$$
\end{definition}

\begin{definition}[Associate function] \label{def:toeplitz_fun}
Let $T_n$ be a Toeplitz matrix according to Definition \ref{def:toeplitz_mat}. A function $f\in \Lone$ is called an associate function to $T_n$ if  $f(z)$ is associated with the corresponding Fourier series of the matrix elements i.e.,  $$f(z) \sim \ \sum_{k=-\infty}^\infty c_k e^{ikz},$$
where 
$$ c_k = \frac{1}{2\pi} \int_{-\pi}^{\pi}e^{-ikz} f(z) dz. $$
The matrix is sometimes denoted $T_n(f)$, and in the infinite case we write $T(f)$.
\end{definition}
The notation is motivated by the fact that the associated function determines the distribution of the eigenvalues of the matrix. To state this, we need the following definition.

\begin{definition}[Equal distribution] \label{def:equal_dist}
Let $M$ be a constant in $\R$ and $0<n\in \mathbb{N}$. Also, let $\{a_\ell^{(n)}\}_{\ell=1}^{n}$ and $\{b_\ell^{(n)}\}_{\ell=1}^{n}$ be sets in $\R$ such that$$|a_\ell^{(n)}|<M,\quad|b_\ell^{(n)}|<M,\quad \mbox{ for all } 1\le\ell \le n$$
We say that $\{a_\ell^{(n)}\}_{\ell=1}^{n}$ and $\{b_\ell^{(n)}\}_{\ell=1}^{n}$ are equally distributed on $[-M,\, M]$ if for any continuous function $$F \, : \, [-M,\, M] \rightarrow \R,$$
we have that
$$\lim_{n\rightarrow\infty}\frac{\sum_{\ell=1}^{n}F(a_\ell^{(n)})-F(b_\ell^{(n)})}{n} = 0 .$$
\end{definition}

\begin{proposition}[Spectrum of a Toeplitz matrix]\label{prop:toeplitz_spectrum}
Let $T_n$ be a Toeplitz matrix with associated function $f$. Also let $\lambda^{(n)}_\ell $,  $\ell=1,2,\dots,n$, be the eigenvalues of $T_n$. Then,
$$ \inf_z f(z) \leq\, \lambda_\ell^{(n)}\leq\sup_z f(z).$$
Moreover, the sets $$ \left\{\lambda_\ell^{(n)}\right\}_{\ell=1}^{n} \quad \text{and} \quad \left\{f\left(\frac{2\pi \ell}{n}-\pi \right)\right\}_{\ell=1}^{n} \ $$ are equally distributed.
\end{proposition}
\begin{proof}
See \cite{szego_monograph} or \cite{ gray}.
\end{proof}

\begin{proposition}[Associate function of the square root inverse]\label{prop:sqrt_inv_toep}
Let $T_n$ be a Toeplitz matrix with associate function $f$ and let $f(z)>0$ for all $-\pi\leq z\leq\pi$. Moreover, let $\sqrt{f}$ be such that its Fourier series has absolutely summable coefficients (that is; $\sqrt{f}$ is in the Wiener class \cite{gray}). Now let $K_n$ be the inverse of the hermitian square root of $T_n$,  $$K_n= ((T_n)^{\frac{1}{2}})^{-1}.$$
Then, $K_n$ is asymptotically Toeplitz in the sense that there exists a Toeplitz matrix $\tau_n$ such that for $D_n = K_n-\tau_n$
$$\lim_{n\rightarrow\infty} \frac{1}{n}\Tr(D_n^*D_n) = 0 . $$
Moreover, the associated function to $K_n$ is given by $1/\sqrt{f(z)}$
\end{proposition}
\begin{proof}
If $\sqrt{f}$ is in the Wiener class, then so is $f$. The proposition now follows from application of Theorem 5.3 (c) and Theorem 5.2 (c) in \cite{gray}.
\end{proof}

\begin{definition}[Circulant matrix]\label{def:circulant_mat}
A matrix $C_n$ is called Circulant if it has the following form
$$C_n = \left( \begin{array}{c c c c c}c_{0} & c_{1} & c_{2} & \dots & c_{n-1}\\
c_{n-1} & c_0 & c_{1} & \dots & c_{n-2}\\
c_{n-2} & c_{n-1} & c_{0} & \dots & c_{n-3}\\
\vdots & \vdots & \vdots & \ddots & \vdots\\
c_{1} & c_{2} & c_{3} & \dots & c_{0}
\end{array}\right) \ .$$
It is thus a special case of a Toeplitz matrix.
\end{definition}

\begin{proposition}[Diagonalizing a Circulant matrix]\label{prop:circulant_fft}
Let $C_n$ be a Circulant matrix.
Then,
$$  C_n = U \Lambda_n U^*, $$
where $\Lambda_n = \diag(\lambda_1,\, \lambda_2,\, \dots,\, \lambda_n)$ and $U$ is the Fourier matrix: $[U]_{m,\ell}=e^{-2\pi i m \ell /n}$ for $m,\, \ell = 0,\,1, \,\dots,\, n-1$.
\end{proposition}
\begin{proof}
See \cite{ gray} and \cite{fast_multip}.
\end{proof}

To be concise and to avoid ambiguities regarding the simulations we also define the following words:
\begin{definition}[Payload bits]\label{def:pay_bits}
A \textit{payload bit} is a bit realized from some input distribution. This bit represents pure data that we want to communicate and it is on sets of these that we compute block-error rates (BLER) and bit-error rates (BER); these are also used in the measure of throughput.
\end{definition}
\begin{definition}[Physical bits]\label{def:phy_bits}
A \textit{physical bit} is the bit physically transfered over the channel. A physical bit represents some type of coded realization of a payload bit. The number of physical bits, denoted $\#\text{physical bits}$, is $$\#\text{payload bits} = \#\text{physical bits} \cdot \text{code rate} . $$
When measuring signal-to-noise ratio (SNR), it will be done with respect to physical bits.
\end{definition}


Now introduce the scalar products
\begin{equation}
\label{h_scalar_product}
	H_{k\ell} = \langle h_k, h_\ell \rangle = \intinf h_k(t)\overline{h_\ell(t)} d t
\end{equation}
and define the matrix $H$ whose elements $[H]_{k\ell} = H_{k\ell}$ for $k,\ell = 1, \ldots, n$.
The matrix $H$ is the Gramian of the functions $h_1, \ldots, h_n$.

\begin{proposition}[Positive Definiteness of the Gramian Matrix]
\label{invgram}
Let $h_1, \ldots, h_n\in \Ltwo$ and consider the Gram matrix
$$
H = \begin{pmatrix}
		H_{11} & H_{12} & \cdots & H_{1n} \\ 
		H_{21} & H_{22} & \cdots & H_{2n} \\
		\vdots  	& \vdots  &  & \vdots \\
		H_{n1} & H_{n2} & \cdots & H_{nn} \\
	\end{pmatrix}
$$
where $H_{k\ell} = \langle h_k, h_\ell \rangle$. Then, $H\succeq 0$. Furthermore, $H$ is invertible if and only if
$h_1, ..., h_n$ are linearly independent.
\end{proposition}
\begin{proof}
See \cite[pp. 12-13]{akhiezer:1993}.
\end{proof}

We will consider Mazo's model \cite{mazo:FTN} consisting of a finite number of modulated $\sinc$ pulses transmitted at a rate that is higher than the Nyquist rate. 
Let $T=\frac{1}{2W}$ and

\begin{equation*}
	g(t) = \sqrt{2W}\cdot \sinc{(2 W t)}   = \sqrt{T}\cdot \frac{\sin{2\pi W t}}{\pi t}.
\end{equation*}
Introduce
$$
h_k(t) = g\lp t - \rho \frac{k-1}{2W}\rp
$$ 
for $\rho \in (0,1]$.
For any real number $\tau$, we have that
$$
	g(t-\tau) \xrightarrow{\mathcal{F}}G(\omega) e^{- i\omega \tau}
$$
and so the spectrum of $g(t-\tau)$ is invariant under (time) shifting since
$$
|G(\omega) e^{- i\omega \tau}| = |G(\omega)|.
$$
Now, let $X(t)$ be a signal consisting of 
$m = \lfloor n/\rho\rfloor$ samples transmitted with a rate that is \textit{faster than Nyquist} (that is $\rho<1$), given by

\begin{equation} \label{eq:X}
	X(t) = \sum_{k=1}^{m} A_k h_k(t)
\end{equation}
where $\{A_k\}$ are real valued random variables. 
It's faster than Nyquist in the sense that the pulses are spaced in-between with a time distance $\rho T$ instead of $T$. Thus, the sample density is $\frac{1}{\rho T}=\frac{2W}{\rho}$ samples/second. 

\begin{lemma}
\label{specind}
The functions 
$$
h_k(t) = g\lp t - \rho \frac{k-1}{2W}\rp ,
$$ 
$k=1, ..., m$, are linearly independent.
\end{lemma}
\begin{proof}
See the appendix.
\end{proof}


\begin{figure}[!t]
\centering
\includegraphics[width=0.48\textwidth, height=0.45\textwidth]{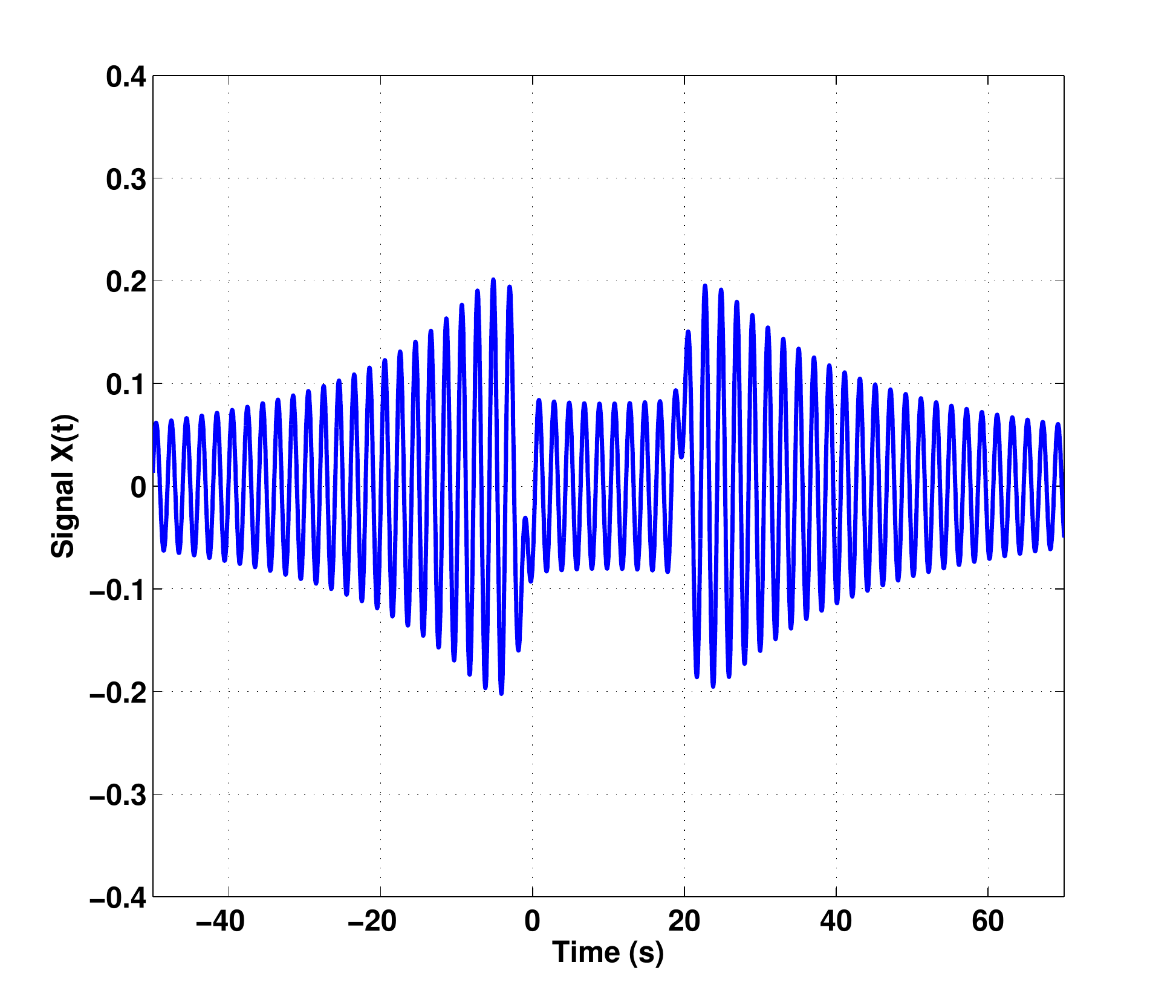}
\caption{Signal of the form \eqref{eq:X} where $\rho=0.81$, $n=20$, and $T=1$ with more than 50\% of the energy is outside the interval $\Omega = [-15, 34]$.} \label{fig:energy_outside_Ex}
\end{figure}

\begin{figure}[!t]
\centering
\includegraphics[width=0.48\textwidth, height=0.45\textwidth]{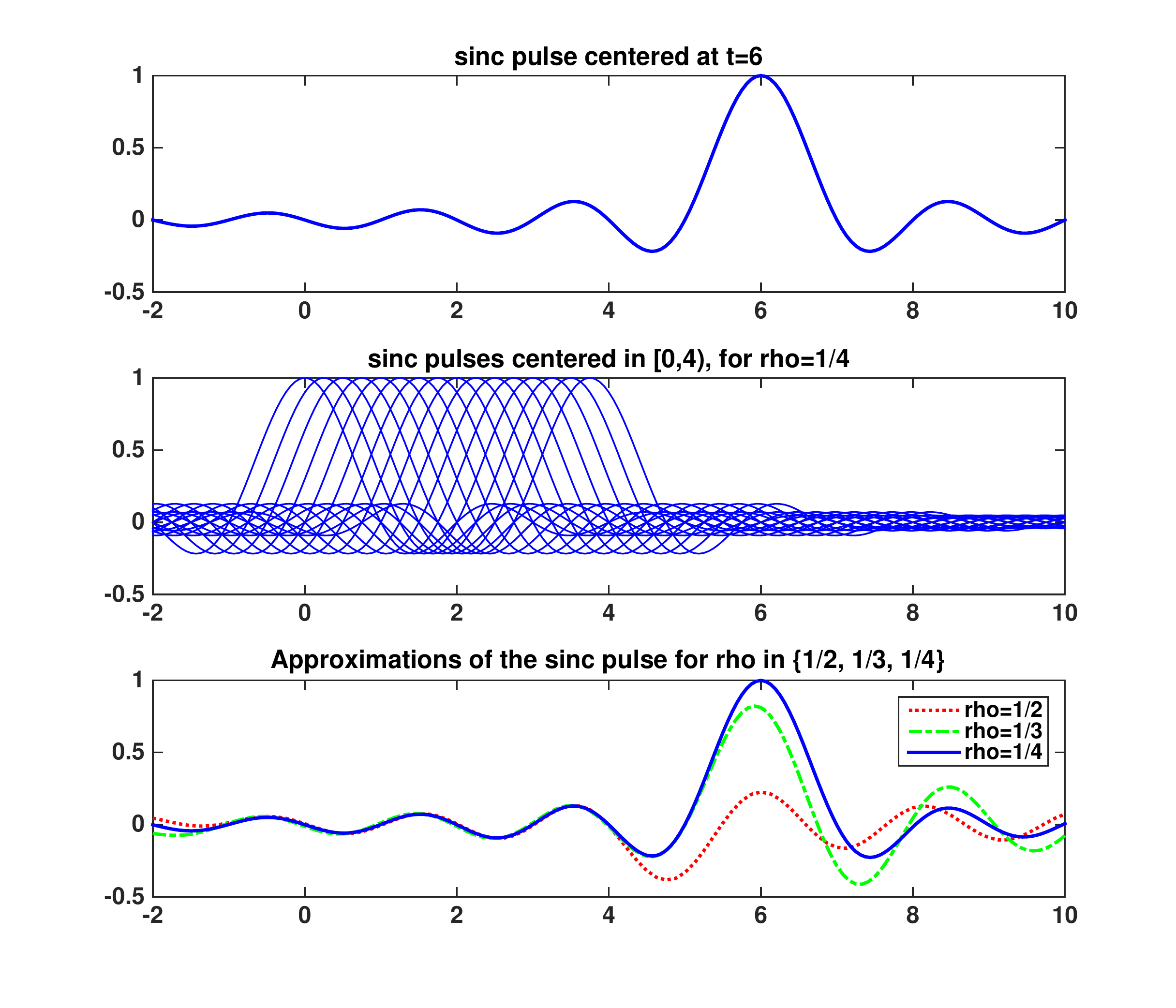}
\caption{Top: sinc pulse centered at $t=6$. Middle: Basis functions for $\rho=1/4$, $n=4$, and $T=1$. Bottom: Optimal $\mathcal{L}_2(\mathbb{R})$ approximations for $\rho\in \{1/2, 1/3, 1/4\}.$} \label{fig:approx2}
\end{figure}

\newcommand{\mR}{{\mathbb R}}
\section{Time Localization and FTN}\label{sec:timeloc}

A signal is well localized to a time interval $\Omega\subset \mR$ if the signal energy outside the interval is small compared to the total signal energy, i.e., if
\begin{equation}\label{eq:TL}
K_{\rm D}(X,\Omega)=\frac{\int_{ \Omega} X(t)^2 dt}{\|X(t)\|^2}\ge 1-\mu
\end{equation} 
where $\mu$ is close to zero. 

An implicit assumption when considering signals $X(t)$ of the form \eqref{eq:X} is that the signal is well localized in time to the interval $\Omega=[0,(n-1)T]$ (that is, approximately time-limited). However, in the faster than Nyquist framework, this assumption is violated and as we will see one can create signals with a considerable proportion of its energy content outside this region.
This implies that signals of the form \eqref{eq:X} with $\rho\ll 1$ may not satisfy \eqref{eq:TL} even for moderate size $\mu$ and intervals $\Omega$ that contains $[0,(n-1)T]$, for example  $\Omega=[-M,(n-1)T+M]$, for some $M>0$ (e.g., as depicted in Figure \ref{fig:energy_outside_Ex} with $M=15$).

To see another example of this, consider pulses consisting of the (non-orthogonal) basis functions  $g(t-\rho (k-1))$ for $k=0,\ldots, m = \lfloor 4/\rho\rfloor$ for a given $\rho<1$ and where $W=1/2$. Figure \ref{fig:approx2} shows the best quadratic approximations of the 
{$\sinc$} pulse $g(t-6)$ for $\rho\in \{1/2, 1/3, 1/4\}$. By Lemma \ref{specind} such approximation is not exact, but the approximation is very close for $\rho=1/4$. This example shows that using the faster than Nyqvist framework with small $\rho$, one can create signals with support outside the designated ``time slot'' of the signal, for example approximating a $\sinc$ pulse centred outside the interval $[0,(n-1)T]$. In fact, it is possible for each fixed $n>0$ to approximate any band limited signal as a faster than Nyqvist pulse if the rate $\rho$ is sufficiently small. This is the content of the following proposition.
 \begin{proposition}\label{thm:L2}
 Let $n>0$ be fixed. 
Any $\mathcal{L}_2(\mR)$ function whose Fourier Transform is restricted to the frequency band $[-W,W]$  can be approximated arbitrary close (in $\mathcal{L}_2(\mR)$) by a function 
 $X(t)$ of the form  \eqref{eq:X} for some positive number $\rho<1$.
 \end{proposition}
\begin{proof}
See the Appendix.
\end{proof}

It is possible to construct examples where the energy is not localized to the interval $[0,n-1]$ also for $\rho$ larger than the Mazo limit. Consider the problem of determining the maximal energy outside the interval $\Omega_m:=[-m,m+n-1]$ for a signal $X(t)$ of the form \eqref{eq:X}, i.e., 
\begin{align*}
\max_{X(t)} & \int_{\mR\setminus \Omega_m} X(t)^2 dt ~~\mbox{ subject to } \|X(t)\|_2\le 1\\
\text{subject to } & X(t) = \sum_{k=1}^{m} A_k h_k(t)
\end{align*}
Note that this optimization problem may be written as an eigenvalue problem
\begin{equation*}
\max_{A\in \mR^{m}} A^\intercal(H-H(\Omega_m))A  \mbox{ subject to } A^\intercal HA\le 1,
\end{equation*} 
where 
\[
H(\Omega_m)_{k,\ell}=\int_{\Omega_m}h_k(t)\overline{h_\ell(t)}dt,
\]
and the maximum corresponds to the largest eigenvalue of the matrix $I-H^{-1}H(\Omega_m)$. Figure \eqref{fig:energy_outside1} shows how the maximal energy outside $\Omega_m$  depends on $m$ where $\rho\in \{0.7, 0.81, 0.9, 1\}$ for the case $n=20$. Figure \eqref{fig:energy_outside_Ex} shows an example where $\rho=0.81$ and more than $50\%$ of the energy is outside the interval $\Omega_{15}=[-15, 34]$.

\begin{figure}[!t]
\centering
\includegraphics[width=0.48\textwidth, height=0.45\textwidth]{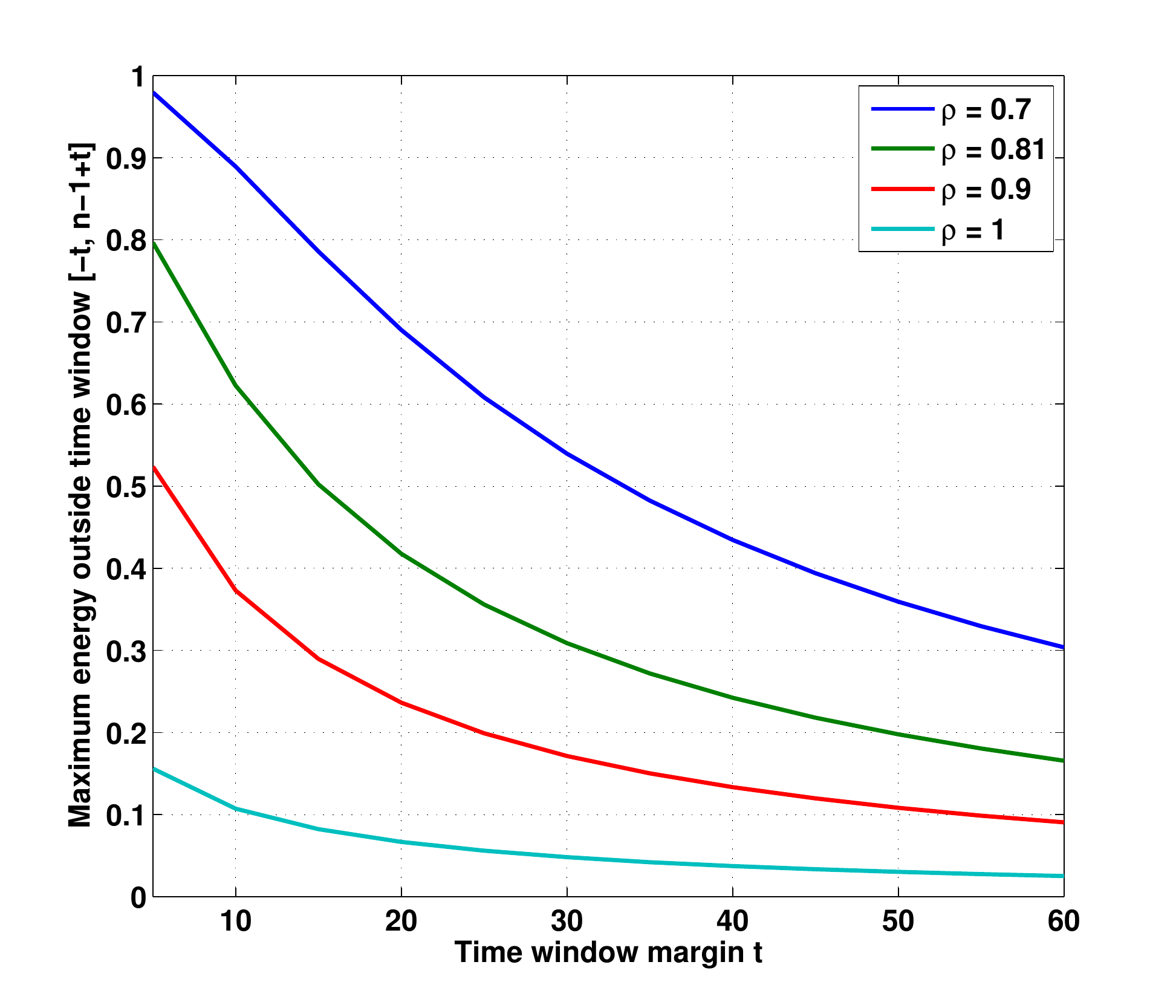}
\caption{ Maximal energy outside $\Omega_m=[-m,m+n-1]$  as a function of $m$, of a signal  \eqref{eq:X} where $\rho\in \{0.7, 0.81, 0.9, 1\}$ and where $n=20$.} \label{fig:energy_outside1}
\end{figure}

This effect may also be seen when the alphabet is constrained to a discrete set of points, e.g., for binary symbols. To illustrate this, let the signal be of the form  \eqref{eq:X} where $A_k\in \{-1, +1\}$ with $T=1$, $\rho=0.9$, and $m=400$. The signal corresponding to $A_k=(-1)^k$ is depicted in  Figure \ref{fig:energy_outside}.
The darker region of the signal, between the vertical bars, is the time $[-\rho, \rho m]$ where the signal energy is supposed to be localized. However, as can be seen a considerable proportion of the energy falls outside this region. In fact numerical integration shows that about $75\%$ of the energy is confined in the interval $t\in[-\rho , \rho m]$, and only about $27\%$ to the interval $t\in[0, \rho (m-1)]$.

\begin{figure}[!t]
\centering
\includegraphics[width=0.48\textwidth, height=0.45\textwidth]{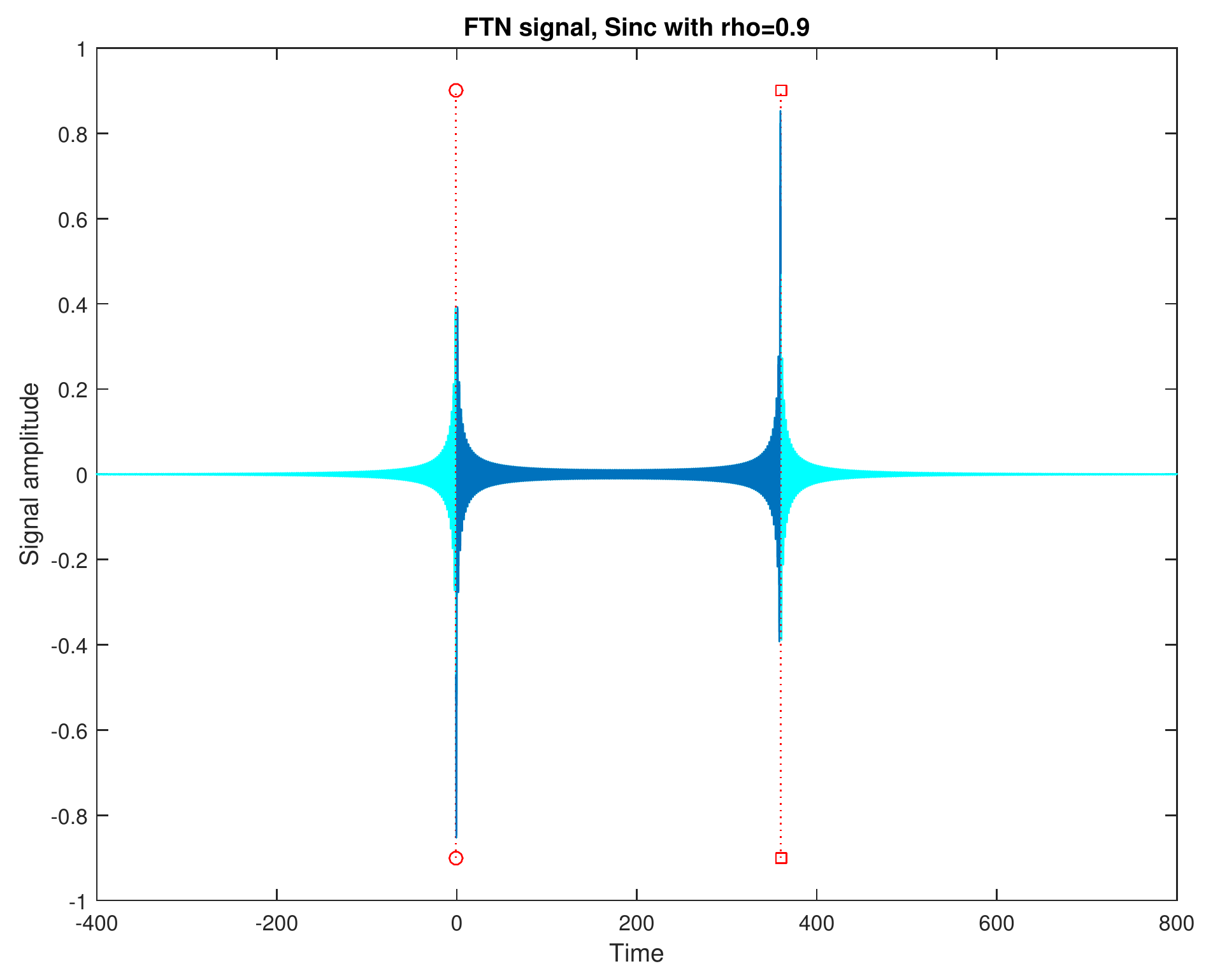}
\caption{This figure shows the resulting signal when a Sinc-pulse is transmitted with a completely alternating sequence of $A_k=(-1)^k$ for $k=1,\ldots, 400$. The darker part of the signal, surrounded by vertical bars, 
is the part $t=-\rho$ and $t=400\rho$ respectively; that is one symbol-time before/after the first/last of the main peaks. } \label{fig:energy_outside}
\end{figure}


\section{Time Localization and capacity of FTN }
\label{timeFTN}

Next, we examine the FTN model from a viewpoint of channel capacity. Assuming that the FTN model \eqref{eq:X} only contains signals that are well localized to the interval $[0,(n-1)T]$, we will see that the Shannon capacity can be violated. Again this illustrates that the FTN model contains signals that are not well localized in time.  

The energy of a signal given by (\ref{eq:X}) is
\begin{align*}
\E\|X\|^2 	&= \E \left\| \sum_{k=1}^{m} A_k h_k \right\|^2\\
			&= \E \left \langle \sum_{k=1}^{m} A_k h_k, \sum_{\ell=1}^{m} A_\ell h_\ell \right \rangle \\
			&=  \E \sum_{k=1}^{m} \sum_{\ell=1}^{m}  H_{k\ell} A_k A_\ell \\
			&= \E ~ A^\intercal H A.
\end{align*}
Let $E_s$ be the energy per sample in the case of $\rho=1$. We will restrict the energy to not exceed the energy for the $Nyquist$ signaling case, that is, the same as the energy for signaling $n$ samples. For $\rho = n/m$, we have
$$
\E ~ A^\intercal H A = m \rho E_s  = \frac{n}{\rho}\rho E_s  = n E_s
$$
which gives an average energy per sample to be
$$
 \frac{1}{m} \E~A^\intercal H A  = \rho E_s. 
$$
Since the time duration between consecutive samples is $\rho T$, the time spanned of $m$ samples is $m \rho T$. 
Now if the non-orthogonal pulses were well localized in time, we would get over an infinite time horizon the average power
\begin{equation}
\label{pow_av}
\lim_{m \rightarrow \infty} \frac{1}{m\rho T} \E ~ A^\intercal H A = \frac{\rho E_s}{\rho T} = P.
\end{equation}
Thus, the average power does not exceed that of the Nyquist signaling case.
The received signal over the AWGN channel is thus given by
$$
Y(t) = X(t) + Z(t)
$$
where $Z(t)$ is a zero mean white noise Gaussian process. We define the measurements
\begin{equation}
\label{nonorthsamp}
\begin{aligned}
Y_k &= \langle h_k, Y \rangle  \\
	 &= \langle h_k, X + Z \rangle \\
	 &= \langle h_k, \sum_{\ell=1}^{m} A_\ell h_\ell \rangle + Z_k \\
	 &= \sum_{\ell=1}^{m}  \langle h_k, h_\ell \rangle A_\ell + Z_k \\
	 &= \sum_{\ell=1}^{m}  H_{k\ell} A_\ell + Z_k.
\end{aligned}
\end{equation}
We may write (\ref{nonorthsamp}) in the more compact form
\begin{equation}
\label{sampMatrix}
Y = H A + Z.
\end{equation}
It's not hard to check that $\{Y_k\}$ is not a sequence of independent variables since $\{h_k\}$ is not an orthogonal set of functions. The covariance of
 $Z_k$ and $Z_\ell$ is given by 
\begin{align}
E\{Z_k Z_\ell\} 	&= \frac{N_0}{2} \langle h_k, h_\ell \rangle =  \frac{N_0}{2} H_{k\ell}. \label{ncov}
\end{align}
The functions $\{h_k\}$ are linearly independent according to Lemma \ref{specind},
and Proposition \ref{invgram} implies that $H\succ 0$. Let $H^{\frac{1}{2}}\succ 0$ be the unique positive definite matrix such that $H^{\frac{1}{2}}\cdot H^{\frac{1}{2}} = H$.
Then, we may write 
$$Z =  H^{\frac{1}{2}} V , ~~~ V\sim \gv(0, \frac{N_0}{2} I_m)$$
and (\ref{sampMatrix}) is equivalent to
\begin{equation}
\label{sampMatrix2}
Y = H A +  H^{\frac{1}{2}} V.
\end{equation}
Since $H$ is positive definite, it's invertible, and so is $H^{\frac{1}{2}}$.
Multiplying the left and right hand side of Equation ($\ref{sampMatrix2}$) 
by $H^{-\frac{1}{2}}$ gives
\begin{equation}
\label{sampMatrix3}
S = H^{-\frac{1}{2}} Y = H^{\frac{1}{2}} A +  V.
\end{equation}
Now let $X =  H^{\frac{1}{2}} A$, using the precoding $A = H^{-\frac{1}{2}} X$. This is always possible since $H$ is nonsingular. The energy constraint on $X$ is then
$$
\sum_{k=1}^m \E |X_k|^2 = \E |X|^2 = \E A^\intercal H A = \rho m E_s.
$$
The channel capacity over the discrete time channel
\begin{equation}\label{eq:sample4}
S_k = X_k + V_k, ~~~ V_k\sim \gv(0, \frac{N_0}{2}), ~~ k = 1, ..., m,
\end{equation}
with a sequence $\{X_k\}$ of total energy
$$
 \sum_{k=1}^{m} \E |X_k|^2  = \rho m E_s
$$
is maximized for a sequence $\{X_k\}$ of
independent identically distributed Gaussian variables with $X_k \sim \gv(0, \rho E_s)$. 
Thus, the capacity
of $m$ samples is 
\begin{align*}
C_m 	&= m\cdot \frac{1}{2} \log_2\lp 1 + \frac{2\rho E_s}{N_0}\rp \\
		&= m\cdot \frac{1}{2} \log_2\lp 1 + \frac{\rho P}{N_0W}\rp ~~ \text{bits}.
\end{align*}
Recall that the sample density is $\frac{2W}{\rho}$ samples/second. Then, the capacity per sample is $C_m/m$ and the average capacity in bits per second as time goes to infinity is
\begin{align*}
C(\rho) 	&= \lim_{m\rightarrow \infty}  \frac{2W}{\rho}  \cdot \frac{C_m}{m} \\
			&=  \frac{W}{\rho}  \log_2\lp 1 + \frac{\rho P}{N_0W}\rp ~~ \text{bits/second}.
\end{align*}

From the derivations above, we see that Mazo's model implies that the channel capacity in fact increases as we pack signals tighter in time, that is when $\rho$ decreases(note that $\rho=1$ renders  Shannon's result for orthogonal pulses). This rather unexpected (and false) result may be explained away by considering the energy localization of signals in time. The signals are not well localized in time and the limit in the expression of the average power in (\ref{pow_av}) is not equal to $P$ but is equal to a smaller number since the energy could be arbitrarily small over a time period $m\rho T$.

\section{FTN for Non Ideal Pulses}
\subsection{FTN Gramian and spectrum}
We will show that  FTN can be applied to make good use of spectrum leakage, when pulses with discontinuous spectrum are difficult to use. 
The following result is similar to results that can be found in \cite{rusek:2009}, but provides additional insight since the associated function is explicitly presented, something we will make use of later. 
\begin{proposition}[Toeplitz- and pulse-spectrum]\label{lemma:toeplitz}
Let $h(t) \in \Ltwo$, and
$$
h_k(t) = h(t - \tau_k) \xrightarrow{\mathcal{F}}  H(\omega) e^{- i\omega \tau_k} ~~~ k = 1, \ldots, n,
$$
with the restriction that $\tau_k = \tau\cdot (k-1)$ for some constant $\tau\in\R^+\backslash\{0\}$. Furthermore, let $H_n$ be such that $[H_n]_{k\ell}=\langle h_k, h_\ell \rangle$
 Then,
\begin{enumerate}[i)]
\item $H_n$ is a Toeplitz matrix $T_n(f)$ \label{lemma:toeplitz:res1}.
\item The associated function $f\in\Lone$ is given by $$f(z) = \frac{1}{\tau} \sum_{\ell=-\infty}^{\infty} \left|H\left(\frac{z+2\pi \ell}{\tau}\right)\right|^2 ,$$ where $H(\omega)$ is the Fourier transform of $h$. This holds for all $z\in[-\pi,\, \pi]$ except possibly at a set of measure zero. \label{lemma:toeplitz:res2}
\end{enumerate}
\end{proposition}
\begin{proof}
See the Appendix.
\end{proof}
What this proposition says is that the Gram matrix generated by a set of systematically shifted pulses is in fact Toeplitz, and the associated function $f$ is the folded spectrum of the pulses.
Given a closed form of the Fourier coefficients $c_{k-1} = [H_n]_{k1}$ we can always try to calculate $f$ directly. However, Proposition \ref{lemma:toeplitz} gives a more systematic and general framework to it.
From this result we can for example see that a transmission that is free from inter-symbol interference needs to have a folded spectrum that is constant, something that is well known in the literature \cite{proakis,goldsmith:wireless,lapidoth}.

\subsection{Precoding and capacity of the precoded signal}
Consider a pulse shape $g(t)\in \Ltwo$ with Fourier transform $G(\omega)$ where $G(\omega)\neq 0$ for $|\omega|\leq 2\pi W'$ and $G(\omega)=0$ for $|\omega|>2\pi W'$. If the transmitted time-shifted pulses $g(t-kT)$ are to be orthogonal, the time shift $T$ must be larger than $T'= \frac{1}{2W'}$, unless $g(t)$ is a $\sinc$ pulse \cite{lapidoth}. However, by using the precoding described in Section \ref{timeFTN}, we can transmit the non-orthogonal pulses $h_k(t) = g(t-kT')$ and get an easy detection algorithm (see Section \ref{timeFTN}). Then, the channel capacity becomes
\begin{align*}
C =  W' \log_2\lp 1 + \frac{P}{N_0W'}\rp ~~ \text{bits/second,}
\end{align*}
which is exactly the capacity for $\sinc$ pulses with average transmit power $P$ and bandwidth $W'$. To use this result we will have to show that 
the signal $X(t) = \sum_{1}^n A_k h_k(t)$ is well localized in time when $n$ is large enough. 

The time localization of (a large number of) $\sinc$ pulses transmitted at the Nyquist rate is well established (as defined and shown in \cite{wyner1966capacity}), and thus, we will show that as the number of pulses goes to infinity, the precoded signal approaches a train of $\sinc$ pulses transmitted at the Nyquist rate.
\begin{theorem}[Time localization of precoded signals]\label{prop:time_localization_precoded}
Let $g(t)\in \Ltwo$ be a given pulse, and let its Fourier transform $G(\omega)$ be such that $|G(\omega)|$ is piecewise continuously differentiable, $G(\omega)\neq 0$ for $|\omega|\leq 2\pi W'$, and $G(\omega)=0$ for $|\omega|> 2\pi W'$. Introduce $$ T' = \frac{1}{2W'} \ , $$
and
$$  h_k(t) =  g(t-kT')\ .$$
Let $H_{2n+1}$ be a $(2n+1)\times (2n+1)$ matrix where $[H_{2n+1}]_{k,l}=\langle h_k, h_\ell \rangle$, $k,\ell =
 -n,\dots,n$,  
$X\in \R^{2n+1}$, $A=H_{2n+1}^{-\frac{1}{2}}X$, and $X(t)$ be given by
$$X(t) = \sum_{k=-n}^{n} A_k h_k(t) \ . $$
 Then, $X(t)$ is well localized in time.
\end{theorem}
\begin{proof}
See the Appendix.
 \end{proof}

\subsection{Complexity and implementation}
In general the problem of estimating the transmitted sequence is NP-hard \cite{verdu} in the case of inter-symbol interference and usually, one has to rely on the Viterbi algorithm which has exponential complexity \cite{proakis}.

However, the precoding scheme that has been suggested in Section~\ref{timeFTN} reduces the problem to an independent estimation for each sample, which is a much easier problem. A block diagram of the process can be found in Figure \ref{fig:LinPrecodeAWGNchannel}. Starting from the transmitter side we take the discrete data $X$ as input. This is precoded using $A=H^{-1/2}X$ and sent through a pulse shaping filter \eqref{eq:X} over the AWGN-channel. On the receiver side this is sampled using a matched filter, these samples are decoded to give the discrete output samples $S$, which according to \eqref{eq:sample4} are decoupled such that $S_k$ only depends on $X_k$.

\begin{figure}[!t]
\centering
\includegraphics[width=0.48\textwidth, height=0.30\textwidth]{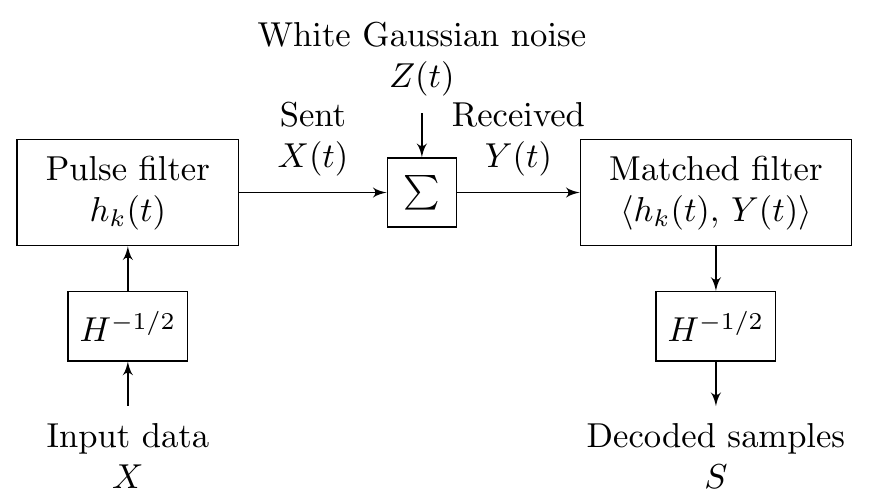}
\caption{Block diagram of the precoding and decoding; with transmission over an AWGN-channel. From independent discrete input data $X$ to independent discrete samples $S$} \label{fig:LinPrecodeAWGNchannel}
\end{figure}

For many cases the algorithm can be implemented efficiently for large $n$ utilizing the fact that the Toeplitz matrix $H$ is asymptotically circulant \cite{gray}. For such cases, $H$ can be approximated by a circulant matrix and the algorithm can be well approximated by an implementation using FFT with complexity $O\left(n\log(n)\right)$. 
Note that if $\rho$ is selected to be too small, then the eigenvalues of 
the Toeplitz matrix $H$ will also be small and this approximate approach will break down.
The implementation used in following subsection is done using both this method and an implementation based on full size matrices, and the results are identical for all the simulations.

\subsection{Example - Root-Raised-Cosine}
We exemplify the theory by applying it to the root-raised-cosine pulse, from this it is also easy to see how it works for other pulses.

As before, let $T=\frac{1}{2W}$, and introduce
\begin{align}\label{eqn:rtrc}
g_{\beta}(t) = \frac{4\beta}{\pi\sqrt{T}} \cdot \frac{ \cos\left( (1+\beta)\pi\frac{t}{T} \right) +\frac{\sin\left( (1-\beta)\pi\frac{t}{T} \right)}{4\beta\frac{t}{T}} }{ 1-\left( 4\beta\frac{t}{T} \right)^2 }\ ,
\end{align}
defined for $\beta\in[0,\,1]$. Let the Fourier transform of $g_{\beta}(t)$ be
\begin{align}\label{eqn:fourier_rtrc}
G_{\beta}(\omega) =\left\{\begin{aligned}
&\sqrt{T} &\quad& 0\leq|\omega|\leq(1-\beta)\frac{\pi}{T}\\
&\text{\eqref{eqn:fourier_rtrc_full_form}} && (1-\beta)\frac{\pi}{T} <|\omega|\leq(1+\beta)\frac{\pi}{T} \\
&0 && |\omega|>(1+\beta)\frac{\pi}{T},
\end{aligned}\right.
\end{align}
with
\begin{align}\label{eqn:fourier_rtrc_full_form}
\sqrt{\frac{T}{2}}\sqrt{1- \sin\left( \frac{T}{2\beta}\left(|\omega| - \frac{\pi}{T}\right) \right)} 
\end{align}

We can see that the pulse is a roll-off pulse that is defined with a spectrum leakage $\beta$, the transform has support in linear frequency on an interval of size $W' = W(1+\beta)$. It is also a Nyquist pulse, since with $T$ as chosen above the pulses are orthogonal.
When applying FTN to the root-raised-cosine we arrive at the following result.
\begin{lemma}[Associated function to the Root-raised cosine Gramian]\label{lemma:toeplitz_spectrum_rtrc}
Let $h_k(t) = \sqrt{\rho} \cdot g_{\beta}\lp t - \rho \frac{k-1}{2W}\rp$, with $g_\beta$ given by \eqref{eqn:rtrc}, $k=1,2,\dots,n$, and $\rho\leq1$. Moreover let $H_n$ be the Gramian of these pulses as given by Proposition \ref{invgram}. Then the associated function $f$ to $H_n$ is given by
\begin{align}\label{eqn:f_rtrc}
&f(z) = \left\{\begin{aligned}
&\eqref{eqn:f_rtrc_leq} && \text{If } (1+\beta) \rho \leq 1& \\
&\eqref{eqn:f_rtrc_geq} && \text{If } (1+\beta) \rho \geq 1,&
\end{aligned}\right.
\end{align}
with
\begin{align}\label{eqn:f_rtrc_leq}
 \left\{
\begin{array}{l l}
0 & z\in [-\pi,\, -(1+\beta)\rho\pi]\\
\frac{1}{2}\left(  1 + \sin\left( \frac{z+\rho\pi}{2\beta\rho} \right)  \right) &  z\in [-(1+\beta)\rho\pi, \, -(1-\beta)\rho\pi]\\
1  & z\in [-(1-\beta)\rho\pi,\, (1-\beta)\rho\pi]\\
\frac{1}{2}\left(  1 - \sin\left( \frac{z-\rho\pi}{2\beta\rho} \right)  \right)  & z\in [(1-\beta)\rho\pi,\, (1+\beta)\rho\pi]\\
0  & z\in [(1+\beta)\rho\pi, \, \pi]\\
\end{array}\right.
\end{align}
and
\begin{align}\label{eqn:f_rtrc_geq}
\left\{
\begin{array}{l}
1-\sin\left( \frac{\pi(1-\rho)}{2\beta\rho} \right) \cos\left( \frac{z+\pi}{2\beta\rho} \right)   \\
\hfill z\in [-\pi,\, -(2-(1+\beta)\rho)\pi]\\
\frac{1}{2}\left(  1 + \sin\left( \frac{z+\rho\pi}{2\beta\rho} \right)  \right)   \\
\qquad\qquad\quad z\in [ -(2-(1+\beta)\rho)\pi, \, -(1-\beta)\rho\pi]\\
1   \\
\hfill z\in [-(1-\beta)\rho\pi,\, (1-\beta)\rho\pi]\\
\frac{1}{2}\left(  1 - \sin\left( \frac{z-\rho\pi}{2\beta\rho} \right)  \right)   \\
\hfill z\in [(1-\beta)\rho\pi,\, (2-(1+\beta)\rho)\pi]\\
1-\sin\left( \frac{\pi(1-\rho)}{2\beta\rho} \right) \cos\left( \frac{z-\pi}{2\beta\rho} \right)   \\
\hfill z\in [(2-(1+\beta)\rho)\pi, \, \pi].\\
\end{array}\right.
\end{align}
\end{lemma}
\begin{proof}
See the Appendix.
\end{proof}

By inspecting Lemma \ref{lemma:toeplitz_spectrum_rtrc} in the context of Proposition \ref{prop:toeplitz_spectrum}, we can see that we only have a numerically well behaved Gramian matrix and numerical stability if $(1+\beta)\rho\geq1$. This is indeed also true for all roll-off pulses that utilize extra bandwidth $\beta$, which can be concluded from Proposition \ref{lemma:toeplitz}.

We can see that with $\rho = \frac{1}{1+\beta}$ and $T' = \rho T$, the root-raised-cosine fulfills all prerequisites in Theorem \ref{prop:time_localization_precoded}. Hence, for a fixed average power $P$ the capacity for precoded FTN with root-raised-cosine becomes
\begin{align*}
C =  (1+\beta)W \log_2\lp 1 + \frac{P}{N_0W(1+\beta)}\rp ~~ \text{bits/second.}
\end{align*}

\subsection{Simulation results}

As a proof of concept, we have conducted simulations where the suggested precoding is applied and implemented, using FFT approximation as described above.
These simulations are based on the model presented in \eqref{sampMatrix2}. 

We apply the root-raised-cosine function as a pulse shape and use a roll-off factor $\beta=0.22$ to mimic existing standards as in \cite{3gpp_rtrc}. 
We are only looking at binary input-output and don't use any higher order modulation.

These simulations are such that we keep the time for a block constant, with the reference that in the Nyquist case ($\rho=1$) one block should be 4000 physical bits. The Nyquist transmission also serves as a reference in the sense that these pulses are of unit energy. We are keeping the power constant for all the schemes, meaning that FTN must use a lower energy per physical bit. This in turn relates to the SNR given in the plots, as this is a Nyquist reference type of SNR in order to be able to make a fair comparison between the different schemes. The SNR is given in dB and relates to the standard deviation of the sampled noise as $$\sigma = 10^{-\text{SNR}_\text{dB}/20} ,$$ since Nyquist transmissions applied unit energy per physical bit. This standard deviation is also used for FTN transmissions regardless of the fact that FTN is using lower energy per physical bit, so the experienced SNR per physical bit will be worse than what is actually written on the axes in the FTN case. This comparison is motivated since it allows us to judge, for a given channel state, if it is worth changing a Nyquist scheme for FTN.

In addition, this simulation also applies (WCDMA) turbo codes according to \cite{3gpp_WCDMA_turbo}. The coding is applied on the payload bits to create the physical bits, $X$, on which we then apply the precoding $H^{-1/2}$. Similarly at the receiver we fist apply FTN decoding, to produce $S$, before using turbo decoding which gives the estimated payload bits. 
For every value of $\rho$ and SNR presented, we have looked at code rates from $1/3$ up to $\approx0.96$ in a grid containing 18 points  (exact code rates are rounded due to finite block length). From this, we have then computed the throughput as:
\begin{align*}
\text{throughput}\,=(1-\text{block-error rate})\cdot\, \#\text{payload bits}  ,
\end{align*}
and selecting the code rate giving highest throughput.
The throughput is given as bits per equivalent time unit and the reference is, as mentioned, the Nyquist scheme transmits $4000$ physical bits, the exact bandwidth is then just a scaling from this.
\begin{figure}[!t]
\centering
\includegraphics[width=0.48\textwidth, height=0.45\textwidth]{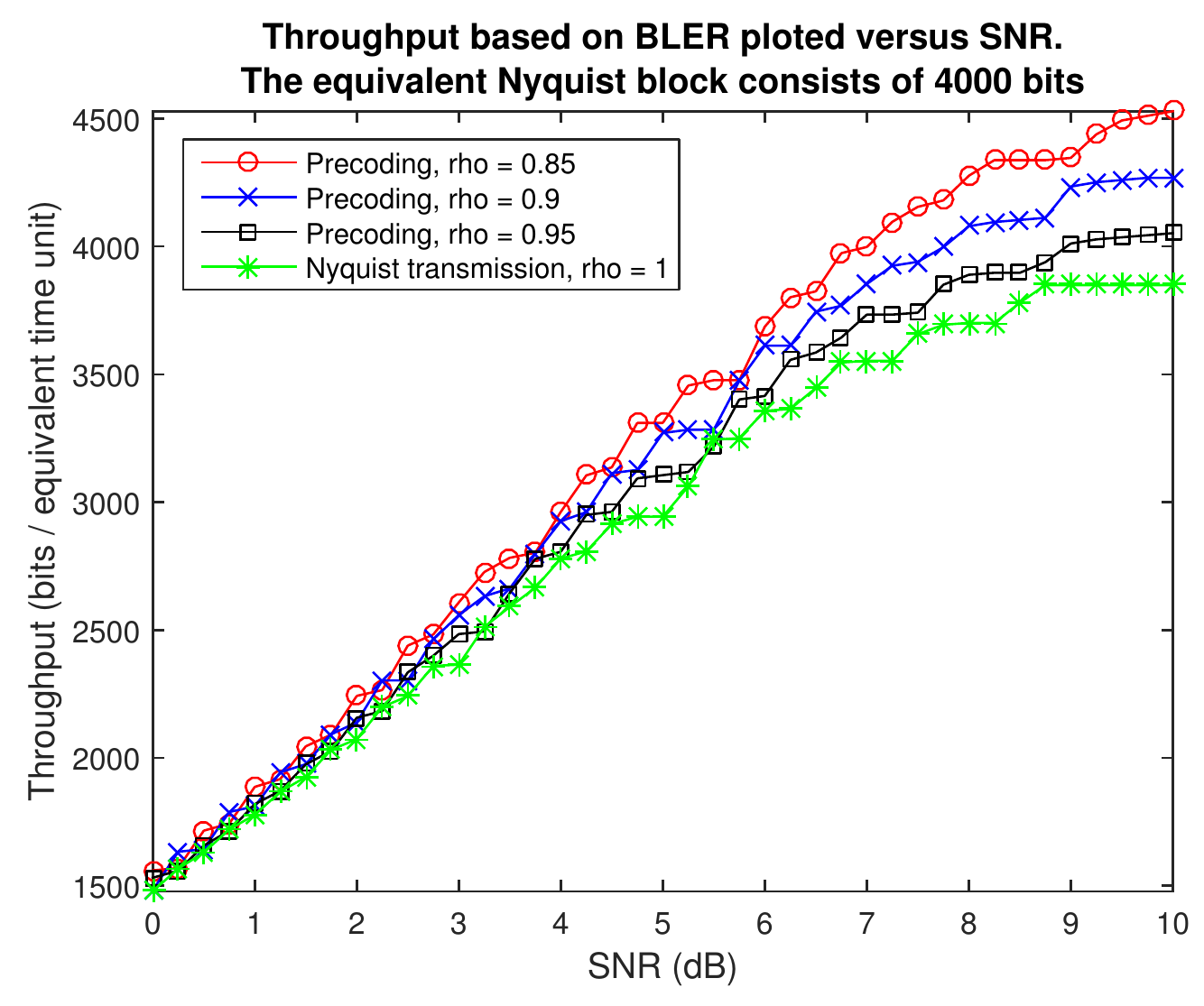}
\caption{Throughput as given by $\text{throughput}\,=(1-\text{block error rate})\cdot\, \#\text{payload bits}$. One can observe that the FTN scheme seems to need a higher SNR in order to reach better throughput, and the lower the $\rho$, the higher the SNR. However, once this is reached it starts to outperform the Nyquist case. The stagnation in throughput at high SNR is due to the code rate not going higher than $\approx0.96$.} \label{fig:BLER_throughput}
\end{figure}

Another way of calculating the throughput, although less attractive for practical purposes, is to rely on the bit-error rate instead of the block-error rate.
The resulting throughput can be found in Figure \ref{fig:BER_throughput}. It should however be noted that this optimization resulted in all schemes using the highest possible code rate for all SNR. 
\begin{figure}[!t]
\centering
\includegraphics[width=0.48\textwidth, height=0.45\textwidth]{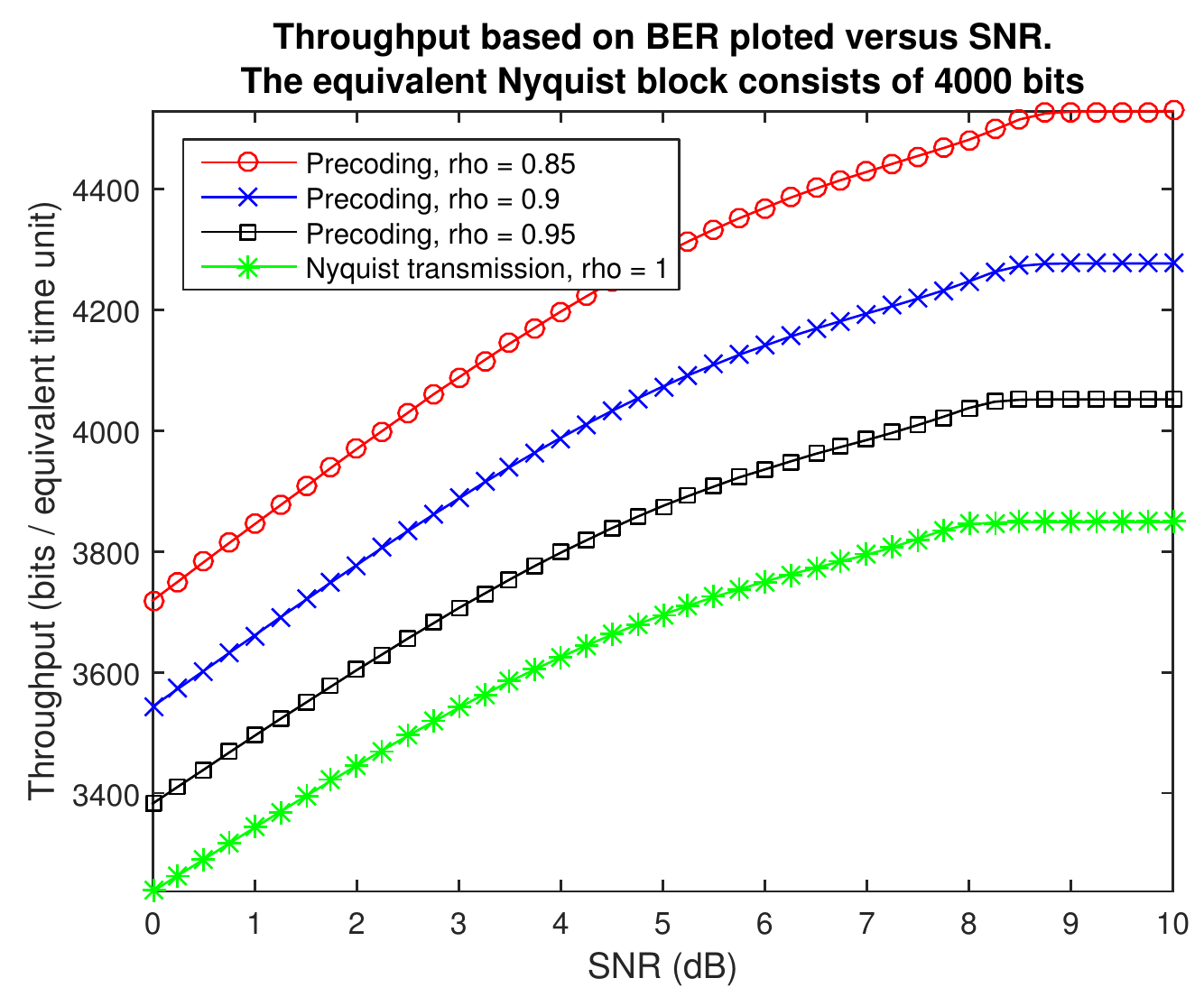}
\caption{Throughput as given by $\text{throughput}\,=(1-\text{bit-error rate})\cdot\, \#\text{payload bits}$. This plot shows throughput gains using FTN even at low SNR, although this is less informative from a system perspective. The same stagnation in throughput at high SNR as in Figure \ref{fig:BLER_throughput} can also be seen here, we have reached $\text{BER}\,=0$.} \label{fig:BER_throughput}
\end{figure}



\section{Conclusions}
We considered communication over the analog white Gaussian noise channel using a finite bandwidth 
$[-W,W]$ and \textit{non-orthogonal} pulses by signaling at a rate that is higher than the Nyquist rate.
We showed that the conclusions in \cite{mazo:FTN}, that one may transmit symbols carried by $\sinc$ pulses at a higher rate than that dictated by Nyquist without loosing in bit error rate don't imply that the bit error rate per \textit{time} unit decreases. This was demonstrated by showing that if the model in \cite{mazo:FTN} is valid to consider bit error rates per time unit, then it means that non-orthongal signals may achieve a capacity for the AWGN channel that is higher than the Shannon capacity. We explain this phenomenon by means of an example where we show that non-orthogonal signals do not give rise to well localized energy in time. Thus, it's not physically correct to talk about bits per \textit{second}, as the energy of non-orthogonal signals may be more spread over time. Therefore, the results of Mazo \cite{mazo:FTN} do not imply that one can transmit more data \textit{per time unit} without degrading performance in terms of error probability.

We also considered FTN signaling in the case of pulses that are different from the $\sinc$ pulses. We showed that one may use a precoding scheme of low complexity, in order
to remove the inter-symbol interference. This leads to the possibility of increasing the number of transmitted samples per time unit and compensate for spectral efficiency losses due to signaling at the Nyquist rate of the non $\sinc$ pulses. Thus we can achieve the Shannon capacity when the same energy is spent on transmitting with the ideal $\sinc$ pulses. 

\section{Acknowledgements}
This work was in part supported by Bitynamics Research, Ericsson Research, and the ACCESS Linnaeus Centre at KTH. 

\bibliography{additional_refs,mybib}


\appendix

\subsection{Proof of Lemma \ref{specind}}
To prove the statement, we need the following result.
\begin{proposition}
\label{indexp}
Let $\Omega \subset \R$ with $\mathbf{int} (\Omega) \neq \emptyset$ and
 $\alpha_1, ..., \alpha_n\in \C$ with $\alpha_i \neq \alpha_j$ for $i\neq j$. 
Then, 
$e^{\alpha_1 \omega }, \ldots, e^{\alpha_n \omega}$ are linearly independent for $\omega \in \Omega$.
\end{proposition}
\begin{proof}
We will prove the theorem by induction over $n$. The case $n = 1$ is trivial.
Start with $n=2$. If $e^{\alpha_1 \omega }$ and $e^{\alpha_2 \omega }$ are linearly dependent, then there exist $c_1, c_2\in \C$
such that
$$
c_1 e^{\alpha_1 \omega }+ c_2 e^{\alpha_2 \omega } = 0 ~~~~ \forall \omega \in \Omega
$$
which implies that
$$
c_1 e^{(\alpha_1 - \alpha_2)\omega }  = -c_2 ~~~~ \forall \omega \in \Omega.
$$
In order for the equality above to hold for more than one value of $\omega$, and since $\alpha_1 \neq \alpha_2$, we
must have 
$c_1 = c_2 = 0.$
Now let $n=k-1$ and suppose that $e^{\alpha_1 \omega }, ..., e^{\alpha_{k-1} \omega}$  are linearly independent but $e^{\alpha_1 \omega }, ..., e^{\alpha_{k} \omega}$ are linearly dependent. Then, there exist 
 $c_1, ..., c_k\in \C$ with $c_k\neq 0$ such that
$$
c_1 e^{\alpha_1 \omega }+ \cdots + c_k e^{\alpha_k \omega } = 0 ~~~~ \forall \omega \in \Omega.
$$
 Multiplying the right and left hand sides of the equality above with $e^{-\alpha_k \omega } $ gives
\begin{equation}
\label{ckeq}
c_1 e^{(\alpha_1 - \alpha_k) \omega }+ \cdots +  c_{k-1}e^{(\alpha_{k-1} - \alpha_k) \omega } + c_k = 0 ~~~~ \forall \omega \in \Omega
\end{equation}
Since $\mathbf{int} (\Omega) \neq \emptyset$, we may take the derivative of the left hand side of (\ref{ckeq}) for $\omega \in \mathbf{int} (\Omega)$.
Taking the derivative of the equality (\ref{ckeq}) gives
\begin{align*}
& (\alpha_1 - \alpha_k) c_1 e^{(\alpha_1 - \alpha_k) \omega } + \cdots +\\
& (\alpha_{k-1} - \alpha_{k}) c_{k-1}e^{(\alpha_{k-1}-\alpha_k) \omega } = 0, ~~~~ \forall \omega \in  \mathbf{int} (\Omega).
\end{align*}
Now multiply the right and left hand sides of the equality above with $e^{\alpha_k \omega } $:
\begin{align*}
& (\alpha_1 - \alpha_k) c_1 e^{\alpha_1  \omega } + \cdots + (\alpha_{k-1} - \alpha_{k}) c_{k-1}e^{\alpha_{k-1} \omega } = 0
\end{align*}
for all $\omega \in \mathbf{int} (\Omega)$. Since $e^{\alpha_1 \omega }, ..., e^{\alpha_{k-1} \omega}$ are linearly independent by the induction hypothesis, we
must have $c_1= \cdots = c_{k-1} = 0$. But this contradicts (\ref{ckeq}) since $c_k\neq 0$. We conclude that our assumption that  $e^{\alpha_1 \omega }, ..., e^{\alpha_{k} \omega}$ are linearly dependent was false, and linear independence for $n=k-1$ implies that for $n=k$. Thus, $e^{\alpha_1 \omega }, ..., e^{\alpha_{n} \omega}$ are lineary independent for any positive integer $n$, and we are done.
\end{proof}


Next, we will prove Lemma \ref{specind} by contradiction. Suppose that $h_1, ..., h_{n}$ are linearly dependent. Then, there exists
$0\neq (c_1, ..., c_n) \in \C^{n}$ such that
$$
c_1 h_1(t) + \cdots c_n h_n(t) = 0, \quad \forall  t \in \R.
$$
Taking the Fourier transform of the equation above, we get
$$
c_1 H(\omega) e^{- i\omega \tau_1} + \cdots + c_nH(\omega) e^{- i\omega \tau_n} = 0, ~~~~ \forall \omega \in \R.
$$
In particular, we have that 
$$
c_1 H(\omega) e^{- i\omega \tau_1} + \cdots + c_nH(\omega) e^{- i\omega \tau_n} = 0, 
$$
for all $\omega \in (-2\pi W, 2\pi W)$, which implies that
$$
c_1 e^{- i\omega \tau_1} + \cdots + c_n e^{- i\omega \tau_n} = 0, ~~~~ \forall  \omega \in (-2\pi W, 2\pi W).
$$
Now letting $\alpha_k = - i \tau_k$, for $k = 1, \ldots, n$, and using Proposition \ref{indexp} with $\Omega=[-2\pi W,2\pi W]$, 
we see that $e^{- i\omega \tau_k}$ are linearly independent so we must have $(c_1, \ldots, c_n)  = 0$, a contradiction.
Therefore $h_1, \ldots, h_n$  must be linearly independent, which concludes the proof.\qed 


\subsection{Proof of Proposition \ref{thm:L2}}

First, note that we may take $T=1$ without loss of generality.
Next, introduce the following two linear projection operators. Let $D:\mathcal{L}_2(\mR)\to \mathcal{L}_2(\mR)$ be the orthogonal projection onto the set of time limited function
\begin{equation*}
(D X)(t)=\begin{cases} X(t)\quad t\in [0,n]\\
0 \hspace{25pt}  \mbox{otherwise},
\end{cases}
\end{equation*}
and let $B:\mathcal{L}_2(\mR)\to \mathcal{L}_2(\mR)$ be the orthogonal projection onto the set of band limited functions with frequency support $[-W,W]$
\begin{equation*}
(BX)(t)=\int_{-\infty}^{\infty}X(\tau)\sinc(2W(t-\tau))d\tau.
\end{equation*}
We now have the following lemma.

\begin{lemma}\label{lm:range} The closure of the range of $BD$ is equal to the range of $B$, i.e., $\overline{\mbox{Range(BD)}}=\mbox{Range(B)}$.
\end{lemma}

\begin{proof}[Proof of Lemma \ref{lm:range}]
Note that $\overline{\mbox{Range(BD)}}\subseteq \mbox{Range(B)}$ since $\mbox{Range(B)}$ is closed. 
To show that $\overline{\mbox{Range(BD)}}\supseteq \mbox{Range}(B)$, assume that $X(t)=\mbox{Range}(B) \cap \overline{\mbox{Range(BD)}}^\perp$.
By construction, we have that 
\begin{equation*}
\langle X,BDY \rangle=0 \quad \mbox{ for all } Y\in \mathcal{L}_2(\mR),
\end{equation*}
and hence 
\begin{equation}\label{eq:x_perp_dy}
\langle BX,DY \rangle=\langle X,DY \rangle=0 \quad \mbox{ for all } Y\in\mathcal{L}_2(\mR),
\end{equation}
since $B$ is an orthogonal projection and is thus self-adjoint. Now \eqref{eq:x_perp_dy} holds only if $X(t)=0$ on $t\in [0,n]$, which is only possible if $X\equiv 0$, since $X(t)$ is band limited. Therefore  $\overline{\mbox{Range(BD)}}\supseteq \mbox{Range}(B)$,
and the lemma is complete.
\end{proof}

By Lemma~\ref{lm:range} we may pick $Y\in \mathcal{L}_2(\mR)$ such that $\hat X=BD Y$ is arbitrary close to the desired band limited function. Without loss of generality such $Y$ may be chosen with support in $[0,n]$. Next, let
\[
A_k=\int_{\rho(k-1)/(2W)}^{\rho k/(2W)} Y(t)dt, \quad \mbox{ for } k=1,\ldots, m=\lfloor n/\rho\rfloor.
\]
With this construction $\sum_{k=1}^{m} A_k \delta(t-\rho\frac{k-1}{2W})\to Y$  weakly as $\rho\to 0$. Since the total variation norm is uniformly bounded $\sum_{k=1}^m |A_k|\le \|Y\|_{\mathcal{L}_1}$ we have that 
$X(t) = \sum_{k=1}^{m} A_k h_k(t)\to \hat X$ as $\rho\to 0$, and the proof is complete.\qed

\subsection*{Proof of Proposition \ref{lemma:toeplitz}}
The first statement is obvious by using $\tau_k = \tau \cdot (k-1)$ and direct calculation of \eqref{h_scalar_product}.

To prove the second statement, we consider
\begin{align*}
[H_n]_{k1} = c_{k-1} = \frac{1}{2\pi}&\int_{-\pi}^{\pi} f(z)e^{-i(k-1)z}d z  ,
\end{align*}
but we also know that
\begin{align*}
[H_n]_{k1} & =  \langle h_{k}, h_1 \rangle
= \intinf h_{k}(t)\overline{h_1(t)} d t \\
&= \frac{1}{2\pi}\intinf H(\omega)e^{-i(k-1)\tau\omega}\overline{H(\omega)} d \omega\\
& = \frac{1}{2\pi}\intinf \frac{1}{\tau}\left|H\left(\frac{z}{\tau}\right)\right|^2 e^{-i(k-1)z} dz .
\end{align*}
The integral is absolutely convergent since $h\in\Ltwo$, which allows us to do some further manipulations of that expression.
\begin{align*}
&\intinf \frac{1}{\tau}\left|H(\frac{z}{\tau})\right|^2 e^{-i(k-1)z} d z \\
&= \sum_{\ell=-\infty}^{\infty} \int_{\pi(2\ell-1)}^{\pi(2\ell+1)} \frac{1}{\tau}\left|H\left(\frac{z}{\tau}\right)\right|^2 e^{-i(k-1)z} d z\\
&= \sum_{\ell=-\infty}^{\infty} \int_{\pi}^{\pi} \frac{1}{\tau}\left|H\left(\frac{z+2\pi \ell}{\tau}\right)\right|^2 e^{-i(k-1)z-i2\pi k\ell} d z\\
&= \int_{\pi}^{\pi} e^{-i(k-1)z} \frac{1}{\tau} \sum_{\ell=-\infty}^{\infty} \left|H\left(\frac{z+2\pi \ell}{\tau}\right)\right|^2 d z .
\end{align*}
This shows that $c_{k-1}$ is the $(k-1)$:th Fourier series coefficient of the function $\frac{1}{\tau} \sum_{\ell=-\infty}^{\infty} \left|H\left(\frac{z+2\pi \ell}{\tau}\right)\right|^2 $, and the result now follows from the uniqueness of Fourier series.

Finally to show that $f\in\Lone$, we use 
Proposition \ref{prop:toeplitz_spectrum} and Proposition \ref{invgram} to get that $f(z)\geq0$.
Hence, $|f(z)|=f(z)$ and
\begin{align*}
\int_{-\pi}^{\pi} \left|f(z)\right|dz = c_0 = \intinf \frac{1}{\tau}\left|H(\frac{z}{\tau})\right|^2 d z < \infty ,
\end{align*}
since $h$ is in $\Ltwo$.
\qed

\begin{remark}
It should be noted that the associated function $f$ is in $\Lone$, and hence the formulation $f(z) = \frac{1}{\tau} \sum_{\ell=-\infty}^{\infty} \left|H\left(\frac{z+2\pi \ell}{\tau}\right)\right|^2 $ is strictly speaking only valid for $z\in[-\pi,\pi]$. However, the sum is such that it provides a periodic extension of $f$ from $[-\pi,\pi]$ to the whole of $\R$.
\end{remark}
\subsection*{Proof of Theorem \ref{prop:time_localization_precoded}}
For ease of notation, we will without loss of generality consider a symmetric transmission, that is 
\begin{align*}
X(t) = \sum_{k=-n}^{n} A_k g(t-kT') \ ,
\end{align*}
where $A$ is the vector of precoded symbols, $A=H_{2n+1}^{-\frac{1}{2}}X$.

First, we know from Proposition \ref{lemma:toeplitz} that $H_{2n+1}$ is Toeplitz with associated function 
$$f(z) = \frac{1}{T'} \sum_{\ell=-\infty}^{\infty} \left|G\left(\frac{z+2\pi \ell}{T'}\right)\right|^2.$$
For our particular choice of $T'$ we have that $G\left(\frac{z+2\pi \ell}{T'}\right) = 0$  for $\ell\neq0$  and $z\in[-\pi,\pi]$. Thus, for $z\in[-\pi,\pi]$, we have
$$f(z) = \frac{1}{T'} G\left(\frac{z}{T'}\right)^2 .$$
It's easy to see that $f(z)$ is $2\pi$-periodic and it's enough to consider the interval $z\in[-\pi,\pi]$. 
Now we will show that $\sqrt{f}$ indeed has an absolutely summable Fourier series, which will allow us to use Proposition \ref{prop:sqrt_inv_toep} in the analysis of $H_{2n+1}^{-\frac{1}{2}}$. We have that

$$\sqrt{f(z)} = \frac{1}{\sqrt{T'}}\left|G\left(\frac{z}{T'}\right)\right|, $$
and since $|G(\omega)|$ is piecewise continuously differentiable, so is $\sqrt{f}$. Thus, we know that $\sqrt{f}$ has a Fourier series whose coefficients converge absolutely \cite[Theorem 7.21]{amann} and hence belongs to the Wiener class.

Propositions \ref{prop:toeplitz_spectrum}, \ref{prop:sqrt_inv_toep}, and \ref{invgram}, together with the definitions of $G(\omega)$ and $T'$ assures us of the existence of $H_{2n+1}^{-\frac{1}{2}}$.  Now let
\begin{align*}
&H_{2n+1}^{-\frac{1}{2}} = \left(\begin{array}{c c c c c c}
c_{1,0} & \dots & c_{1,n} & \dots & c_{1,2n}\\
\vdots & \dots & \vdots & \dots & \vdots\\
c_{n+1,-n} & \dots &  c_{n+1,0} &  \dots & c_{n+1,n} \\
\vdots & \dots & \vdots & \dots & \vdots\\
c_{2n+1, -2n} & \dots &  c_{2n+1,-n} & \dots & c_{2n+1,0}\\
\end{array}\right) ^\intercal
\end{align*}
Although the notation might indicate otherwise, the matrix is actually symmetric and the reason for numbering its transpose will be apparent later. By Proposition \ref{prop:sqrt_inv_toep} the matrix is asymptotically Toeplitz, and tends towards
\begin{align*}
&H^{-\frac{1}{2}} = \left(\begin{array}{c c c c c c c }
\ddots & \ddots & \ddots & \ddots & \ddots & \ddots & \ddots\\
\ddots & c_{-1} & c_{0} &  c_{1} &  c_{2} & c_{3} & \ddots \\
\ddots & c_{-2} & c_{-1} &  c_{0} &  c_{1} & c_{2} & \ddots \\
\ddots & c_{-3} & c_{-2} &  c_{-1} &  c_{0} & c_{1} & \ddots\\
\ddots & \ddots & \ddots & \ddots & \ddots & \ddots & \ddots\\
\end{array}\right)
\end{align*}
with associate function $\hat{f}$, as $n\rightarrow\infty$. We have already established that $\sqrt{f}$ fulfils the conditions of Proposition \ref{lemma:toeplitz} and thus $\hat{f} = 1/\sqrt{f}$.
The signal $X(t)$ can now be written as 
\begin{align}\label{eqn:X(t)_vec_form}
X(t) =X^\intercal \left(H_{2n+1}^{-\frac{1}{2}}\right)^\intercal \left(\begin{array}{c}
g(t+nT')\\
g(t+(n-1)T')\\
\vdots\\
g(t)\\
\vdots\\
g(t-(n-1)T')\\
g(t-nT')
\end{array}\right) \ .
\end{align}
By inspecting \eqref{eqn:X(t)_vec_form}, we can see that each uncoded symbol $X_\ell$ will be carried by an effective pulse $\xi_\ell^{(2n+1)}(t)$ such that
\begin{align*}
X(t) = \sum_{\ell=-n}^n X_\ell \, \xi_\ell^{(2n+1)}(t).
\end{align*}
The pulse $\xi_\ell^{(2n+1)}(t)$ is given by
\begin{align*}
\xi_l^{(2n+1)}(t)=\sum_{k=-n}^{n} c_{l+n+1,k-l} \ g(t-kT') ,
\end{align*}
or equivalently
\begin{align*}
\xi_l^{(2n+1)}(t)=\sum_{k=-n-l}^{n-l} c_{l+n+1,k} \ g(t-(k-l)T') .
\end{align*}
Taking the Fourier transform of $\xi_l^{(2n+1)}(t)$ gives
\begin{align*}
\Xi_\ell^{(2n+1)}(\omega)=G(\omega) e^{i\ell\omega T'}  \sum_{k=-n-l}^{n-l} c_{n+1,k} \ e^{-ik\omega T'}  .
\end{align*}
Taking the limit $n\rightarrow\infty$, we get 
\begin{align*}
\Xi_\ell^{(2n+1)}(\omega) \xrightarrow[n\rightarrow\infty]{ } \Xi_\ell(\omega),
\end{align*}
where
\begin{align*}
\Xi_\ell(\omega) & = G(\omega) e^{i\ell\omega T'} \sum_{k=-\infty}^{\infty} c_{k} e^{-ik\omega T'}\\
&= G(\omega) e^{i\ell\omega T'}  \hat{f}(\omega T') \\
& = \frac{G(\omega)e^{i\ell\omega T'} }{\frac{1}{\sqrt{T'}}G\left(\omega\right)} .
\end{align*}
We conclude that the spectrum of $\xi_l^{(2n+1)}(t)$ as $n$ goes to infinity is given by
\begin{align*}
 \Xi_0(\omega)
&= \left\{ \begin{array}{c l}
\sqrt{T'}e^{i\ell\omega T'}  & \left|\omega\right|\leq 2\pi W'\\
0 & \left|\omega\right|> 2\pi W',
\end{array}\right.
\end{align*}
 of which the envelope is the spectrum of a $\sinc$ pulse with linear bandwidth $2W'$. Since the pulses are transmitted with a time spacing of $T' = \frac{1}{2W'}$ seconds and since we know that a signal consisting of $\sinc$ pulses of bandwidth $2W'$ and time shift spacing $T' = \frac{1}{2W'}$ are well localized in time, we conclude that the signal $X(t)$ is well localized in time.
 \qed

\subsection*{Proof of Lemma \ref{lemma:toeplitz_spectrum_rtrc}}
The lemma is proved by direct calculations. The assumption $\rho\leq1$ is a technical one which ensures $(1-\beta)\rho\geq2-(1+\beta)\rho$, and with $\beta\in[0,\,1]$ it also implies $(1+\beta)\rho\leq2$.

We have that the Gramian for the root-raised-cosine pulses is simply given by a raised cosine with roll-off $\beta$ and that is sampled at integer instances and with $T=1/\rho$; namely
\begin{align*}
\begin{aligned}[]
[H_n]_{m\ell} =\rho\sinc\left(\rho(m-\ell)\right)\frac{\cos\left(\pi\beta\rho(m-\ell)\right)}{1-4\beta^2\rho^2(m-\ell)^2}.
\end{aligned}
\end{align*}
 This can in turn be rewritten as
\begin{align*}
\begin{aligned}[]
&[H_{n}]_{m\ell}=\\
&\frac{(1+\beta)}{2} \cdot \frac{\rho}{1 - 4\beta^2\rho^2\left(m-\ell\right)^2}  \cdot \sinc((1+\beta)\rho(m-\ell)) \\
&+ \frac{(1 -\beta)}{2}\cdot \frac{\rho}{1 - 4\beta^2\rho^2\left(m-\ell\right)^2 } \cdot \sinc((1 - \beta)\rho(m-\ell)).
\end{aligned}
\end{align*}
We will use one of Parseval's identities for periodic convolution, namely
\begin{align}\label{eqn:parseval}
f(z)
= \sum_{k=-\infty}^{\infty}c_k^1 c_k^2 e^{ikz} 
= \frac{1}{2\pi}\int_{-\pi}^{\pi} g_1(t)g_2(z-t) \dd t.
\end{align}
Thus we identify the sought associate function $f$ as the convolution of the two functions $g_1(z)$ and $g_2(z)$ with the corresponding Fourier series coefficients:
\begin{align}\label{eqn:c1_rtrc}
\begin{aligned}
c_k^1 = &
\frac{(1+\beta)\rho}{2}\cdot \sinc((1+\beta)\rho k) + \\
&\frac{(1 -\beta)\rho}{2}\cdot  \sinc((1 - \beta)\rho k)
\end{aligned}
\end{align}
and
\begin{align}\label{eqn:c2_rtrc}
c_k^2 = 
\frac{-1}{4\beta^2\rho^2k^2-1}  .
\end{align}

We can then based on \eqref{eqn:c1_rtrc} conclude that $g_1$ is a sum of rectangular functions, periodic on the interval $[-\pi,\, \pi]$ and looks like:
\begin{align}\label{eqn:g1_rtrc}
g_1(z) = \left\{
\begin{array}{l l}
\eqref{eqn:g1_rtrc_leq} & \text{ if } (1+\beta)\rho<1\\
\eqref{eqn:g1_rtrc_geq} & \text{ if } (1+\beta)\rho\geq1  ,\\
\end{array} \right.
\end{align}
with
\begin{align}\label{eqn:g1_rtrc_leq}
\frac{1}{2} \rect\left(\frac{z}{2(1+\beta)\rho\pi}\right)
+ \frac{1}{2}\rect\left(\frac{z}{2(1-\beta)\rho\pi}\right)  ,
\end{align}
and
\begin{align}\label{eqn:g1_rtrc_geq}
\begin{aligned}
&1 -\frac{1}{2} \rect\left(\frac{z}{2(2-(1+\beta)\rho)\pi}\right)\\
&+\frac{1}{2} \rect\left(\frac{z}{2(1-\beta)\rho\pi}\right)  .
\end{aligned}
\end{align}
The function $g_1$ is hence different depending on the value of the parameter combination $(1+\beta)\rho$.

Evaluating $g_2(z)$ we try to find the Fourier series of the following function
\begin{align}
\tilde{g}_2(z) = a \cdot \sin\left( \frac{|z|-d}{2b} \right) \qquad z\in [-\pi,\, \pi]  ,
\end{align}
and find that it has the following coefficients:
\begin{align}\label{eqn:c2*_rtrc}
\begin{aligned}
\tilde{c}^{2}_k &= \frac{-1}{4b^2k^2-1} \\
&\cdot\left(\frac{2ab}{\pi}\cos\left(\frac{d}{2b}\right) + \frac{2ab(-1)^n}{\pi} \cos\left( \frac{\pi-d}{2b} \right) \right)  .
\end{aligned}
\end{align}
Comparing $\tilde{c}^{2}$ with $c^2$, that is \eqref{eqn:c2*_rtrc} with \eqref{eqn:c2_rtrc}, we can identify that in our case we have $b=\beta\rho$, $d=(1-\beta\rho)\pi$, and $a=\frac{\pi}{2\beta\rho}\cdot\frac{1}{\cos((1-\beta\rho)\pi/(2\beta\rho))}$. Making some rearrangements we arrive at
\begin{align}\label{eqn:g2_rtrc}
 g_2(z) = \frac{\pi}{2\beta\rho}\cdot \frac{1}{\sin\left(\frac{\pi}{2\beta\rho}\right)} \cdot \cos\left(\frac{|z|-\pi}{2\beta\rho}\right)  .
\end{align}

Now we can use Parseval's identity, \eqref{eqn:parseval}, to retrieve $f$.
Getting a closed form expression for the convolution is made easier since $g_1$ is constant over intervals. However since $g_1$ changes depending on the value of $(1+\beta)\rho$ and $g_2$ contains and absolute value $|z-t|$ where $t$ is the integration variable the problem grows combinatorially. As an intermediate step we solve
\begin{align}\label{eqn:temp_integral}
&\int_{a}^{b} \cos\left(\frac{|z-t|-\pi}{2\beta\rho}\right)  \dd t =
\left\{\begin{aligned} \eqref{eqn:integral_part1} \quad &z\geq b\\
\eqref{eqn:integral_part2} \quad &  a<z<b\\
\eqref{eqn:integral_part3} \quad & z\leq a
\end{aligned}
\right.  .
\end{align}
With
\begin{align}\label{eqn:integral_part1}
2\beta\rho\left(  -\sin\left( \frac{z-\pi-b}{2\beta\rho}\right)+\sin\left( \frac{z-\pi-a}{2\beta\rho}\right)  \right)
\end{align}
\begin{align}\label{eqn:integral_part2}
\begin{aligned}
2\beta\rho\Bigg(  \sin\left( \frac{-z-\pi+b}{2\beta\rho}\right)+\sin\left( \frac{z-\pi-a}{2\beta\rho}\right) \\
+2\sin\left( \frac{\pi}{2\beta\rho}\right)  \, \Bigg) 
\end{aligned}
\end{align}
\begin{align}\label{eqn:integral_part3}
2\beta\rho\left(  \sin\left( \frac{-z-\pi+b}{2\beta\rho}\right)-\sin\left( \frac{-z-\pi+a}{2\beta\rho}\right)  \right).
\end{align}
Combining the results \eqref{eqn:g1_rtrc}, \eqref{eqn:g2_rtrc} and \eqref{eqn:temp_integral}, and doing some simplifications to the expressions gives the sought associate function found in \eqref{eqn:f_rtrc}. \qed

\end{document}